\documentclass[10pt,final,reqno,conference, a4paper]{IEEEtran}

\usepackage{amsmath}
\usepackage[algoruled, linesnumbered]{algorithm2e}
\usepackage{amsthm}
\usepackage[dvipsnames]{xcolor}
\usepackage{graphicx}
\usepackage{caption}

\DeclareMathOperator*{\SNR}{SNR}
\DeclareMathOperator*{\argmax}{argmax}

\newtheorem{theorem}{Theorem}
\theoremstyle{definition}
\newtheorem{definition}{Definition}

\begin{document}
\title{Multihop Relaying in Millimeter Wave Networks: A Proportionally Fair Cooperative Network Formation Game}
\IEEEoverridecommandlockouts

\author{%
\IEEEauthorblockN{
Nof Abuzainab\IEEEauthorrefmark{1} and
Corinne Touati\IEEEauthorrefmark{1}\IEEEauthorrefmark{2}}\\
\thanks{This research was supported in part by the French National Research Agency project NETLEARN (ANR–13–INFR–004)}
\IEEEauthorblockA{\IEEEauthorrefmark{1}Inria \quad
\IEEEauthorrefmark{2} CNRS, LIG, Univ. Grenoble Alpes\\
Email: \{nof.abuzainab, corinne.touati\}@inria.fr}
}

\maketitle

\IEEEpeerreviewmaketitle

\begin{abstract}
Millimeter wave channels suffer from considerable degradation in the channel quality when the signal is Non Line of Sight (NLOS) between the source and the destination. Multihop relaying is thus anticipated to improve the communication between a source and its destination. This is achieved by transmitting the signal to a sequence of relays in which a Line of Sight (LOS) signal exists between two nodes along the path, or more generally when the signal is better than the transmitted signal directly from the source to the destination. In this paper, we consider a millimeter wave network composed of multiple source-destination pairs and a set of deployed relays. We formulate the problem of multihop relaying as a cooperative network formation game in which each relay chooses which source-destination pair to assist in order to improve the end-to-end performance, that is, the multihop delay between the source and the destination. Further, we present an algorithm based on the Nash Bargaining Solution to ensure fairness among the different source-destination pairs and assess its efficiency on numerical simulations.
\end{abstract}

\section{Introduction}

The interest in Millimeter Wave Communications has been tremendously increased as a viable technology for fifth generation wireless cellular systems. This is due to the fact that millimeter wave communications support the very high data rates necessary for broadband and multimedia communications thanks to the availability of large bandwidth at the high frequencies. However, communications at these high frequencies suffer from two main drawbacks. The first is that the millimeter wave signal suffers from severe pathloss. To overcome this, there is an active research going on designing beamforming techniques in order to extend the signal range and enable communication between the targeted transmitter and receiver \cite{mmbeam1}-\cite{mmbeam4}. The second drawback is that the millimeter wave signal gets severely attenuated in the case of Non Line of Sight (NLOS) \cite{RRE14},\cite{mmchallenge2}.
To improve communications in case of NLOS, the use of intermediate relays that have LOS (or in general a better) signal with the source, the destination or among each others is suggested. Hence constructing a path between the source and destination using those relays improves the source-destination communication. In this paper, we focus on the second challenge and attempt to design a multihop relaying technique for a millimeter wave networks.

Relaying in general is a well studied topic in wireless communications, and there is a vast literature covering multihop relaying. However, multihop relaying in the context of millimeter wave communications is yet a new topic, and there has been still few works that deal with this issue. The work in \cite{mmscheduling} selects for a given source-destination pair employing millimeter wave RF the best relay within the beamwidth of the source to assist in the transmission in case of Non Line of Sight (NLOS) between the source and the destination. Also based on this relay selection mechanism, a scheduling algorithm is presented for the case when multiple source-destination pairs are present. In \cite{mmmultihop} a centralized algorithm for multihop relaying routing that takes into account the characteristics of the millimeter wave transmissions is presented. In particular, the presence of multiple source-destination pairs, where each source is interested in video streaming to its destination is assumed. Further, the performance is measured in terms of a differentiated quality function of each flow. The algorithm then finds a feasible route of relays for each source-destination pair, and the objective is to maximize the sum of differentiated quality functions for all flows. 

Our work considers the multihop routing problem for multiple source-destination pairs employing millimeter wave RF as in \cite{mmmultihop}. However, our approach is different because we formulate the multihop relaying problem as a \emph{network formation game}. Network formation games have been recently used for multihop relaying in wireless networks (see \cite{NFtree, NFCR1}) but not yet in particular for millimeter wave networks. In \cite{NFtree}, an algorithm based on network formation game is presented that constructs a uplink multicast tree of relays, to which the mobiles can connect to in order to communicate with the base station. In this game, the relays are the players, and their objective is to connect to the tree in such a way that maximizes their utilities, where the utility is measured in terms of per hop delay and bit error rate. In \cite{NFCR1}, a network formation game is formulated for multihop relaying in Cognitive Radio Networks. The game has also a Stackelberg approach in which the primary source-destination pairs are considered as the leaders, and the secondary users are considered as the followers, and the objective is to construct a path of secondary users (that act as relays) between each source-destination pair so as to improve its transmission and to eventually give the secondary users chance of channel access. In both problems, the network formation games are non cooperative i.e. the player moves are based on maximizing their individual utilities.

In contrast, our approach is based on \emph{cooperative} network formation games and in particular a coalitional graph game in which the path between each source destination pair is constructed in a distributed fashion i.e. each relay decides on joining the path of a certain source-destination pair. Each group of relays along the same path forms a coalition. But as opposed to all previously mentioned works, our algorithm achieves \emph{proportional fairness} in order to maintain an acceptable quality for every source-destination pair and to ensure fairness in relay assignment among the different source destination pairs.

\section{System Model}

We consider a set $\mathcal{M}$  of source-destination pairs $\{(s_i,d_i)\}$, ($i=1,2,...,M$, $M=|\mathcal{M}|$) where each source $s_i$ has a file of $B_i$ bits to deliver to destination $d_i$, and a set $\mathcal{N}$ of deployed relays. It is assumed that the nodes employ millimeter wave RF.
Then, the received power $P_R$ is given by
$$
 P_R=AM_TM_Rd^{-\alpha}P_T,
$$
 where $M_T$ and $M_R$ are the antenna gains at the transmitting and receiving nodes respectively, $A$ and $\alpha$ are the pathloss coefficient and exponent, $P_T$ is the transmitted power, and $d$ is the distance between the transmitting and receiving nodes.

Millimeter wave signals get severely attenuated with distance. Hence, we assume that all nodes employ directional beamforming, and that each pair of communicating nodes engage in a beamstearing algorithm in order to achieve the maximum directivity gains. It is further assumed that the beamwidth is very small as transmitting at very high frequencies permits very narrow beamforming. (Some current products such as \cite{AMW} demonstrate that the beamwidth can be as small as 2 degrees and that interference can be eliminated even with nodes along the same path.) This makes it very unlikely for two pairs of nodes to interfere with each other and therefore interference is neglected. Additive white Gaussian noise with zero mean and power spectral density $N_0$ is assumed to be present at each node. Hence, the received Signal to Noise Ratio (SNR) is given by $\SNR = P_r / N_0$
and we assume that the achieved rate $R$ is related to the SNR through Shanon's capacity formula i.e. $$
R=W\log(1+\SNR) = W\log\left(1+\frac{AM_TM_Rd^{-\alpha}P_T}{N_0}\right)
$$ where $W$ is the available bandwidth.

Further, millimeter wave signals can get severely attenuated with blockage, and thus the signal can get considerably degraded in the case of non line of sight (NLOS). Hence, the channel quality between any pair of nodes is dependent whether a line of sight (LOS) signal exists or not. In particular, measurements (such as in \cite{RRE14}) have shown that different pathloss models exist for the LOS and the NLOS cases. We define $A_N$ and $\alpha_N$ to be the pathloss coefficient and exponent respectively for the NLOS case and $A_L$ and $\alpha_L$ to be the pathloss coefficient and exponent for the LOS case.

Since the direct channel between each source-destination may be NLOS, the achieved rate using direct transmission may be low and incur significant delay to deliver file from the source and the destination. The objective is then to devise multihop relaying  i.e. to try to find a path between each source-destination pair using relays (as some relays might have LOS signal with the source and destination and among each other) so as to improve the communication between each source-destination. The performance is assessed by computing the multihop delay i.e. the time spent through the path to deliver the file from the source to its destination. In order to find a path for each source-destination pair, we design an algorithm using a cooperative network formation game or more specifically a coalition graph game in which each relay chooses to connect to one source-destination pair. Our algorithm also ensures proportional fairness among the source-destination pairs.

 \section{Coalition Graph Game Formulation}

 \subsection{Problem Formulation}

We formulate our problem as a cooperative network formation game or more specifically a coalition graph game, where the players are the relays. Each relay chooses to assist one source-destination pair by connecting itself along the path between the chosen source and destination in a way that achieves the best performance possible. Thus, a group of relays assisting the same source-destination pair is considered as a coalition.
This leads us to the following definitions:


\begin{definition}[Path]
A path $\mathcal{P}_i$ between source $s_i$ and destination $d_i$ is a sequence $\sigma_0$, $\sigma_1$,...,$\sigma_{N_i}$,$\sigma_{N_i+1}$, where $N_i$ is the number of relays in the path, $\sigma_0=s_i$ and $\sigma_{N_i+1}=d_i$ and $\sigma_1$,...,$\sigma_{N_i}$ are the relays along $\mathcal{P}_i$. In other words, it is the set of edges given by $\mathcal{P}_i=\{<\sigma_j, \sigma_j+1>,  0 \leq j \leq N_i \}$.
\end{definition}


\begin{definition}[Action Set]
Each relay $r$ (which is presently either unused or assisting a source-destination pair $(s_i, d_i)$) can decide to perform action $a_k$ and assist source-destination $(s_k,d_k)$ by inserting itself between two consecutive nodes along the path of ($s_k,d_k$) in a way that achieves the minimum possible multihop delay. In other words, if $\mathcal{P}_k$ is the current path between $s_k$ and $d_k$, $\mathcal{P}_{k}(r,j)$ is the path formed by inserting relay $R$ between nodes $j$ and $j+1$ along the path: $$\mathcal{P}_{k}(R,j)\hspace{-2pt}=\hspace{-2pt}\left(\mathcal{P}_k\backslash \{<\sigma_j, \sigma_{j+1}>\} \right) \hspace{-1pt}\cup\hspace{-1pt} \{\hspace{-1pt}<\sigma_j, r>,<r, \sigma_{j+1}>\hspace{-1pt}\}.$$ Relay $r$ inserts itself between nodes $\sigma_{j^*}$ and $\sigma_{j^*+1}$ such that $\displaystyle j^*= \argmax_j D_k(\mathcal{P}_{k}(r,j))$ where $D_k(\mathcal{P}_{k}(r,j))$ is the multihop delay along path $\mathcal{P}_{k}(r,j)$. We denote by $\mathcal{P}^*_{k}(r)$ the resulting path. Also, we define action $a_0$ where the relay decides not to assist any pair. The action set for each relay is then $\mathcal{A}=\{a_k, 0 \leq k \leq M\}$.
\end{definition}

\subsection{Proportional Fairness Maximization}
We are interested in allocating the relays to the source-destination pairs in a fair way so as to avoid situations in which all relays would be allocated to a few source-destination pairs (that have better channels with the relays) enjoying very enhanced performance, while other source-destination pairs would not be adequately assisted by the relays and get poor performance.
One approach in cooperative game theory is the Nash Bargaining Solution (NBS)~\cite{NBS}. In the NBS, 
the objective is to choose the strategies of the players that maximize the following objective function, commonly known as the Nash product:
$ \displaystyle \max_{s \in \mathcal{S}}\prod_{i=1}^{N}(U_i(s)-a_i) $,
where $N$ is the number of players and $U_i(s)$ the utility for player $i$ when all players take actions represented by vector $s$ and vector $a=(a_1,a_2,...,a_N)$ is known as the disagreement point, where each $a_i$ corresponds to the utility value of player $i$ when no agreement is reached.
More often, the value of the $d_i$s are assumed to be zero. Using this assumption and by taking the logarithm of the objective function we get our objective function:
$
\displaystyle \max_{s \in \mathcal{S}}\sum_{i=1}^{N}\text{log}(U_i(s)) 
$
which corresponds to the proportional fairness~\cite{KMT}.

\begin{definition}[Coalition Value]
The value (or utility) of the coalition of each source-destination pair ($s_i,d_i$) is expressed in terms of the multihop delay, i.e. the time required to deliver the file from $s_i$ to $d_i$, which we denote $D_i(\mathcal{P}_i)$, with $\mathcal{P}_i$ the path corresponding to the relays in the coalition $C_i$.  Hence, in order to minimize the proportional fair sum of $D_i$s, we define the value of coalition $C_i$ to be $V(C_i)=-\text{log}(D_i(\mathcal{P}_i)).$ Then, maximizing the sum of the values of the coalitions amounts to maximizing the proportional fairness of the utilities of the source-destination users with their utilities being proportional to their transfer rates $U_i(\mathcal{P}) = 1 / D_i(\mathcal{P}_i)$.
\end{definition}

We consider the multihop delay as being the sum of delays of all edges from $s_i$ to $d_i$. Hence, it is given by the following expression\begin{equation}
D_i(\mathcal{P}_i) = B_i\sum_{j=0}^{N_i}\frac{1}{R_{\sigma_j,\sigma_{j+1}}},
\label{mhdelay}\end{equation}
where $N_i$ is the number of relays along the path, $B_i$ is the size of the file to transfer, $R_{\sigma_j,\sigma_{j+1}}$ is the rate achieved between node $\sigma_j$ and node $\sigma_{j+1}$. Again, we assume that node $\sigma_0$ is the source $s_i$ and node $\sigma_{N_i+1}$ is the destination $d_i$. Any other node $\sigma_j$ ($1 \leq j \leq N_i$) is the $j^{th}$ relay along the path. Note that in the delay expression of Equation~\ref{mhdelay}, it is assumed that each relay decodes the whole file before transmitting it to the next relay along the path. The multihop delay can be improved in the case where the file is divided into packets and transmissions occur packet by packet. However, the expression is more complicated to handle and relies on assumptions on the packet based system (see, e.g. \cite{NFtree} and \cite{NFCR1}, in which packets arrive at each source at a certain rate, and the average delay at each hop is computed based on modeling the packet service system as $M/G/1$ queue). Hence, our expression constitutes a simple upper bound on the delay for packet based transmissions.

The following theorem shows how the optimal actions of the relays maximize the proportional fairness sum.

\begin{theorem}
If each relay $r$ chooses to connect to the source-destination pair $(s_i,d_i)$ that maximizes its marginal contribution (i.e. that has the maximum $\log(D_i(\mathcal{P}_i))-\log(D_i(\mathcal{P}_{i}(r)))$ if it is positive and to chooses not to assist any pair - i.e. to remain unused - otherwise), the corresponding equilibria are the maximizers of the proportional fair sum.
\end{theorem}

\begin{proof}
We consider the game of transferable utility in which the utility (welfare) of each relay in a coalition is proportional to the collective contribution of all relays in the coalition i.e. the coalition value. Hence, when relay $r$ is assisting source-destination $(s_i,d_i)$, the coalition value is divided equally among the $N_i$ relays in the coalition: $u_r(C_i)=\frac{V(C_i)}{N_i}$.
We also set the utility of any unused relay to be zero (as if they were belonging to a dummy path with null coalition value).

Now, we define the repercussion utility of relay $r$ in coalition $C_i$ as
$$\displaystyle r_r(C_i\cup r)=u_r(C_i \cup r)-\sum_{k=1, k\neq r}^{N_i} (u_k(C_i)- u_k(C_i \cup r)).$$

By substituting the values of the utilities $u_r(C_i\cup r)$,
$u_k(C_i \cup r)$, and $u_k(C_i)$ into $r_r(C_i \cup r)$, we get:
$$\begin{array}{ll}
r_r(C_i) &=u_r(C_i \cup r)-\sum_{k=1, k\neq r}^{N_i} (u_k(C_i)- u_k(C_i \cup r))\\
&=\frac{V(C_i \cup r)}{N_i}-\sum_{k=1, k\neq r}^{N_i}\bigg(\frac{V(C_i)}{N_i-1}-\frac{V(C_i \cup r)}{N_i}\bigg)\\
&=\frac{V(C_i \cup r)}{N_i}-(N_i-1)\left(\frac{V(C_i)}{N_i-1}+ \frac{V(C_i \cup r)}{N_i}\right)\\
&=V(C_i \cup r)-V(C_i)\\
&=\text{log}(D_i(\mathcal{P}_i))-\text{log}(D_i(\mathcal{P}_{i}(r))).
\end{array}
$$

Recall that we assume that the utility of each unused relay is zero. The importance of this assumption is to prevent the relay to join a path that it would harm. This happens in the case when all its repercussion utilities for all paths are negative.

Due to our assumption that the network is interference free, the value of each coalition is not dependent on the other coalitions. It has been proven in \cite{gibbs} that a coalition game that satisfies this property and where repercussion utilities are used is an exact potential game with the sum of the original utilities as the potential function. Hence, our game is a potential game where the potential function is the negative of the proportional fair sum of delays of all source-destination pairs.
The result in Theorem 1 follows since an exact potential game has the property that (at least) one pure Nash equilibrium exists and that the Nash equilibria are the local maximizers of the potential function.
\end{proof}

\subsection{Algorithm}

We present in Algorithm \ref{alg1} a distributed algorithm for our cooperative network formation game in which the relays select to connect to a particular source-destination $(s_i,d_i)$ based on proportional fairness. We assume that the relays have full knowledge of the network topology, and that they store the current value of the multihop delay as well as the current path of each source-destination pair. We assume the size of the broadcasted messages is small (i.e. do not require high bit rates) and thus omnidirectional transmission is used during this phase and that a round-robin algorithm is chosen to select each relay $r$ periodically.

Note that we introduce some randomness as we allow each relay to take some non-optimal decision. Indeed, each relay joins the path that yields the maximum repercussion utility with probability $1-\varepsilon$, where $\varepsilon$ is commonly known as the mutation probability \cite{mutation}. Otherwise, the relay will randomly join the path of any other source-destination pair.

\begin{algorithm}
\Repeat{convergence}{
\ForEach{Relay $r$ in $\mathcal{N}$}{
 \ForEach{source-destination $(s_k,d_k)$ in $\mathcal{M}$} {
 $r$ computes the repercussion utility of connecting to $(s_k,d_k)$:
 $r_k(r)=\log(D_k)-\log(D_{k}(r))$\\
 Let $k = \argmax \log(D_k)-\log(D_{k}(r))$
 }
with probability $1-\varepsilon$:\\
$\{$ Relay $r$ connects to  $(s_k,d_k)$

 Update path $\mathcal{P}_k$

 Relay $r$ broadcasts the updated path $\mathcal{P}_k$ and the new multihop delay $D_{k}(r)$ to all other relays$\}$

 otherwise:

$\{$ $r$ connects randomly to $(s_j,d_j)$ ($j \neq k$)\\
   Update path $\mathcal{P}_j$\\
    Relay $r$ broadcasts the updated path $\mathcal{P}_j$ and the new multihop delay $D_{j}(r)$ to all other relays $\}$\\
 \If{$r$ was previously connected to different $(s_i,d_i)$}{
 Update path $\mathcal{P}_i$ by removing $r$\\
 $r$ broadcasts the updated path $\mathcal{P}_i$ and the new multihop delay for $(s_i,d_i)$
 }
 }
}
\caption{Multihop Relaying Algorithm}\label{alg1}
\vspace{-0.2em}
\end{algorithm}

\vspace{-1em}\subsection{Convergence}
Due to the mutation probability, the evolution of paths of all source-destination pairs forms a Markov chain which is irreducible and aperiodic. Hence, it has a unique stationary distribution. It is shown in \cite{mutation} that  as  $\varepsilon$ tends to zero,
the process converges to a unique limiting distribution. Also, since our game is a finite exact potential game, it admits one or several pure Nash equilibria that are the local maximizers of the Nash product. Hence, as $\varepsilon$ tends to zero, the algorithm converges to a deterministic Nash equilibrium. In order to reach the global maximum of the Nash product, it is useful to incorporate Gibbs Sampling techniques (such as the algorithm in \cite{gibbs}). The drawback of Gibbs Sampling is that the convergence time might be unacceptably large for some scenarios.

\section{Numerical Simulations}
In order to simulate  our multihop relaying algorithm, we consider $M=3$ source-destination pairs and $N=10$ deployed relays. The coordinates of all nodes are generated randomly according to a uniform distribution on a $1000\times1000$ meters rectangular grid.
We set all direct channels between source-destination pairs to be NLOS. All other links are chosen LOS or NLOS ran\-domly. The probability of LOS, $p_L$ is taken according to the statistical blockage model of \cite{mmcoverage}: $p_L=e^{-\beta d}$, where $d$ is the distance between the nodes and $\beta$ is a parameter known as the average LOS range of the network and is related to the density and average blockage sizes and set to $\frac{1}{\beta}=141.4$ meters.
As for the pathloss models for both cases of LOS and NLOS, we use the the values obtained from the measurements in \cite{mmmeasurements}, i.e. $\alpha_N=3.88$ and  $\alpha_L=2.20$ for the pathloss exponents, $A_N=A_L=1$ for the pathloss coefficients and $M_T=M_R=4$ for all antennas gains. The transmission power for all nodes is set to be $P_T=1$ Watts and the AWGN variance is set to $N_0 = -40.87$ dBm. The bandwidth is set to $W=1$ GHz and the size of all files is $B_1=B_2=B_3=1$ Gb. For the algorithm, we choose the mutation probability $\varepsilon=10^{-4}$.

\begin{figure}[!b]
\vspace{-2em}
\centering
\includegraphics[width=8.6 cm,height=6cm,angle=0]{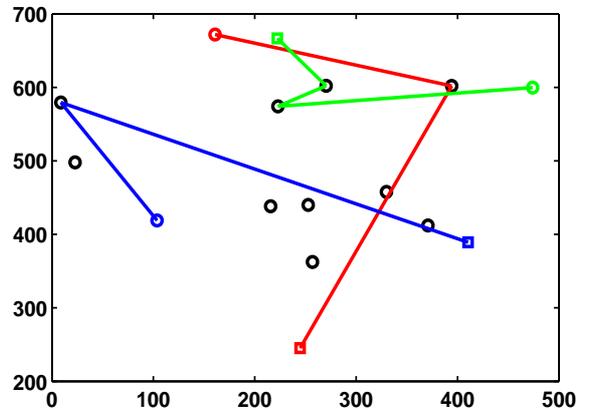}
\caption{Paths formed by Algorithm~\ref{alg1}
\label{multihop1}}
\vspace{-2em}
\end{figure}
\begin{figure}[tb]
\vspace{-1.8em}
\centering
\includegraphics[width=8.6 cm,height=6cm,angle=0]{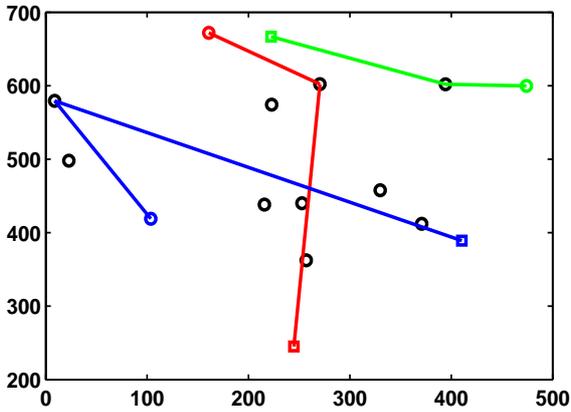}
\caption{Paths formed by the minimum delay algorithm.
\label{multihop2}}
\vspace{-1.7em}
\end{figure}

In order to investigate the potential benefits of proportional fairness, we compare its results to a modified version of the algorithm in which each relay joins the path of the source-destination pair that has the minimum delay. This modified version can be interpreted as a greedy approach whose convergence points are the Nash equilibria of the system.
Figures \ref{multihop1} and \ref{multihop2} show the paths formed between each source destination pair by using our multihop relaying algorithm and the modified minimum delay algorithm respectively. The red, green, and blue circles represent source nodes 1,2, and 3 respectively. The red, green, and blue squares represent destination nodes 1,2, and 3 respectively, and the black circles represent the relay nodes. In this run, most of the formed edges are LOS. Table \ref{delay1} shows the delay values computed for the cases of direct transmission, the proportional fair multihop relaying algorithm, and the modified minimum delay algorithm respectively. First, it is easy to see that the delay values have significantly decreased when multihop relaying (for both cases of proportional fairness and minimum delay) is employed, which confirms the benefits of the multihop relaying algorithm in improving the transmission of the different source destination pairs. Also by comparing the delay values obtained from the two algorithms, we find that when using the minimum delay algorithm, the delay of $(s_2,d_2)$ has slightly dropped from 0.2505 sec (for the proportional fairness algorithm case) to 0.1058 sec while the delay of $(s_1,d_1)$ has increased from 0.1428 sec (for the proportional fairness case) to 1.2451 sec. 
Further, we compute the variances of the delays of all paths for both the proportional fairness algorithm and the minimum delay algorithm. We find out the the variance is 0.282 for the minimum delay algorithm while it is 0.0028 for the proportional fairness algorithm. This shows that proportional fairness can provide a better distribution of the relays among the source destination pairs.

\begin{table}[htb]
\centering
\begin{tabular}{ |l|c|c|c| }
  \hline
  &Direct & Prop Fairness & Minimum Delay \\ \hline
$(s_1,d_1)$  &88.77 & 0.1428 & 1.2451\\
  \hline
  $(s_2,d_2)$  &12.6 & 0.2505 & 0.1058 \\
  \hline
  $(s_3,d_3)$  &23.7 & 0.1318 &  0.1318 \\ \hline\end{tabular}
  \caption{Delay Values (in seconds)}
\label{delay1}
\vspace{-1em}
\end{table}

Further, we run both both algorithms for a thousand times. In each simulation, we randomly generate the coordinates of the nodes and choose which links are LOS. 
After each run, we record the sum of delays of all paths for both algorithms. Then, we compute the average sum of delays for each algorithm out of the 1000 runs. We find that the average is 6.8 sec for our algorithm while it is 9.2 sec for the minimum delay algorithm and 86 sec when no multihoping is used (i.e. through the direct path). These values demonstrate the power of proportional fairness compared to the minimum delay approach: while the minimum delay approach can slightly benefit to some users, it does so at the cost of a decreased global performance. Both techniques exhibit excellent global performance compared to the direct transmission. Further, although it converges to a local optima of the Nash product, we observe very good performance metrics. Hence, while Gibbs sampling techniques could ensure convergence to global optima, it is anticipated that the cost of the convergence time will not be compensated by significant performance improvement.

\section{Conclusion}
We have considered a cooperative network formation algorithm in order to construct a multihop path through relays to improve the transmission of source-destination pairs employing millimeter wave RF. Also, we have considered proportional fairness in our algorithm. Due to the assumption of negligible interference in millimeter wave networks, we could show that our network formation game can be turned into a potential game, whose Nash equilibria maximize the proportional fair sum of the transmission rates. We further proposed a distributed algorithm and assessed its performance by numerical simulations. The results show the considerable performance improvement brought by multihop relaying especially in the case where the sources and destinations are NLOS. They also confirm the benefits of proportional fairness in achieving balanced allocations among the source-destination pairs while maintaining a good overall performance of the system.

\enlargethispage{2\baselineskip}



\end{document}